\documentclass{llncs}

\usepackage[utf8]{inputenc}
\usepackage{amsmath}
\usepackage{amssymb}
\usepackage{graphicx}
\usepackage{paralist}
\usepackage{subcaption}

\newif\ifarxiv
\arxivtrue
\usepackage[english]{babel}

\usepackage[pdfpagelabels,colorlinks,citecolor=blue,linkcolor=blue,urlcolor=blue]{hyperref}

\graphicspath{{images/}}

\def\spread[#1,#2]{#1_1, #1_2, \dots, #1_{#2} } % easy to make: A_1, A_2, ..., A_n

\newcommand{\TRI}{\bigtriangleup}
\def\G{\mathcal{G}}
\DeclareMathOperator\dist{\rm dist}

\begin{document}

% ============================================================
\title{On the edge-length ratio of $2$-trees\thanks{The research was initiated during workshop Homonolo 2018. Research partially supported by MIUR, the Italian Ministry of Education, University and Research, under Grant 20174LF3T8 AHeAD: efficient Algorithms for HArnessing networked Data.
V.~Bla{\v z}ej acknowledges the support of the OP VVV MEYS funded project CZ.02.1.01/0.0/0.0/16\_019/0000765 ``Research Center for Informatics''.
This work was supported by the Grant Agency of the Czech Technical University in Prague, grant No. SGS20/208/OHK3/3T/18.
The work of J.~Fiala was supported by the grant 19-17314J of the GA \v{C}R.
}}

\author{V{\'a}clav Bla{\v z}ej\inst{1}\orcidID{0000-0001-9165-6280} \and
Ji{\v r}{\'\i} Fiala\inst{2}\orcidID{0000-0002-8108-567X} \and
Giuseppe Liotta\inst{3}\orcidID{0000-0002-2886-9694}}%end author

% there seems to be an issue with runningheads llncs allignment
\authorrunning{\qquad\qquad\qquad\qquad\quad\,\,\, V. Bla{\v z}ej \and J. Fiala \and G. Liotta}
% \titlerunning{test}

\institute{Faculty of Information Technology, Czech Technical University in Prague, Czech Republic\\
\email{vaclav.blazej@fit.cvut.cz} \and
Faculty of Mathematics and Physics, Charles University, Prague, Czech Republic\\
\email{fiala@kam.mff.cuni.cz} \and
Dipartimento di Ingegneria, Universit\`a degli Studi di Perugia, Italy\\
\email{giuseppe.liotta@unipg.it}}
%end institute

% ============================================================
\maketitle
% ============================================================

\begin{abstract}
We study planar straight-line drawings of graphs that minimize the ratio between the length of the longest and the shortest edge. We answer a question of Lazard et al. [Theor. Comput. Sci. {\bf 770} (2019), 88--94] and, for any given constant $r$, we provide a $2$-tree which does not admit a planar straight-line drawing with a ratio bounded by $r$. When the ratio is restricted to adjacent edges only, we prove that any $2$-tree admits a planar straight-line drawing whose edge-length ratio is at most $4 + \varepsilon$ for any arbitrarily small $\varepsilon > 0$, hence the upper bound on the local edge-length ratio of partial $2$-trees is $4$.

\end{abstract}

\section{Introduction}\label{se:intro}
Straight-line drawings of planar graphs are thoroughly studied both for their theoretical interest and their applications in a variety of disciplines (see, e.g.,~\cite{DBLP:books/ph/BattistaETT99,DBLP:reference/crc/2013gd}). Different quality measures for planar straight-line drawings have been considered in the literature, including area, angular resolution, slope number, average edge length, and total edge length (see, e.g.,~\cite{DBLP:reference/crc/GiacomoLT17,HoffmannKKR14,DBLP:books/ws/NishizekiR04}).

This paper studies the problem of computing planar straight-line drawings of graphs where the length ratio of the longest to the shortest edge is as small as possible.
We recall that the problem of deciding whether a graph admits a planar straight-line drawing with specified edge lengths is NP-complete even when restricted to $3$-connected planar graphs~\cite{EadesW90} and the completeness persists in the case when all given lengths are equal~\cite{CabelloDR07}. In addition, deciding whether a degree-4 tree has a planar drawing such that all edges have the same length and the vertices are at integer grid points is NP-complete~\cite{DBLP:journals/ipl/BhattC87}.

In the attempt of relaxing the edge length conditions which make the problem hard, Hoffmann et al.~\cite{HoffmannKKR14} propose to minimize the ratio between the longest and the shortest edges among all straight-line drawings of a graph. While the problem remains hard for general graphs (through approximation of unit disk graphs~\cite{ChenJKXZ11}), Lazard et al. prove~\cite{LazardLL19} that any outerplanar graph admits a planar straight-line drawing such that the length ratio of the longest to the shortest edges is strictly less than $2$. This result is tight in the sense that for any $\varepsilon >0$ there are outerplanar graphs that cannot be drawn with an edge-length ratio smaller than $2-\varepsilon$. Lazard et al. also ask whether their construction could be extended to the class of series-parallel graphs.

We answer this question in the negative sense, by showing that a subclass of series-parallel graphs, called $2$-trees, does not allow any planar straight-line drawing of bounded edge-length ratio.
In fact, a corollary of our main result is the existence of an $\Omega(\log n)$ lower bound for the edge-length ratio of planar straight-line drawings of $n$-vertex $2$-trees.
Motivated by this negative result, we consider a local measure of edge-length ratio and prove that when the ratio is restricted only to the adjacent edges, any series-parallel graph admits a planar straight-line drawing with local edge-length ratio at most $4 + \varepsilon$, for any arbitrarily small $\varepsilon >0$. The proof of this upper bound is constructive, and it gives rise to a linear-time algorithm assuming a real RAM model of computation.

It is worth noticing that Borrazzo and Frati have shown that any $2$-tree on $n$ vertices can be drawn with edge-length ratio $O(n^{0.695})$~\cite{BorrazzoF20}. This, together with our $\Omega(\log n)$ result, defines a non-trivial gap between the upper and lower bound on the edge-length ratio of planar straight-line drawings of partial $2$-trees. We recall that Borrazzo and Frati also show an $\Omega(n)$ lower bound on the edge-length ratio of general planar graphs~\cite{BorrazzoF20}.

The rest of the paper is organized as follows. Preliminaries are in Section~\ref{se:preli}; the $\Omega(\log n)$ lower bound is proved in Section~\ref{se:lower-bound}; Section~\ref{se:local} presents a constructive argument for an upper bound on the local edge-length ratio of partial $2$-trees. Conclusions and open problems can be found in Section~\ref{se:open}.
\ifarxiv
Omitted proof can be found in the appendix.
\else
Omitted proof can be found in the appendix of the full version of the paper which is available online~\cite{DBLP:journals/corr/abs-1909-11152}.
\fi

\section{Preliminaries}\label{se:preli}

We use capital letters $A,B,\dots$, for the points in the Euclidean plane.
For points $A$ and $B$, let $|AB|$ denote the Euclidean distance between $A$ and $B$.
The symbol $\TRI ABC$ denotes the triangle determined by three distinct non-colinear points $A$, $B$, and $C$. The symbol $\angle BAC$ stands for the angle at vertex A of the triangle $\TRI ABC$.

For a polygon $Q$, we denote its perimeter by $P(Q)$ and its area by $A(Q)$.

We consider finite nonempty planar graphs and their planar straight-line drawings.
Once a straight-line drawing of a graph $G$ is given,
with a slight abuse of notation we use the same symbol for a vertex $U$ and the point $U$ representing the vertex $U$ in the drawing; the same symbol $UV$ for an edge and the corresponding segment;
as well as $\TRI UVW$ for an induced cycle of length three and the corresponding triangle.

When we consider graphs as combinatorial objects, we often use lowercase symbols $u$ or $e$ for the vertices and edges.

%For a polygon $Q$, the symbols $P(Q)$ and $A(Q)$ denote the perimeter and area of $Q$, respectively.

The \emph{edge-length ratio} of a planar straight-line drawing of a graph $G$
is the ratio between the length of the longest and the shortest edge of the drawing.

\begin{definition}
The \emph{edge-length ratio} $\rho(G)$ of a planar graph $G$
is the infimum edge-length ratio taken over all planar straight-line drawings of $G$.
\end{definition}

The class of $2$-trees is defined recursively: an edge is a $2$-tree. If $e$ is an edge of a $2$-tree, then the graph, formed by adding a new vertex $u$ adjacent to both endpoints of $e$, is also a $2$-tree. In such a situation we say that $u$ has been added as a \emph{simplicial} vertex to $e$. A \emph{partial $2$-tree} is a subgraph of a $2$-tree.

\iffalse
A vertex is called \emph{simplicial} when its neighborhood forms a clique.
A complete graph on $k+1$ vertices is a \emph{$k$-tree}; a graph constructed from a $k$-tree by adding a simplicial vertex to a clique of size $k$ is also a $k$-tree. A \emph{partial $k$-tree} is a subgraph of a $k$-tree.
\todo[inline]{R3: Line 67-70 This general definition is not used in the paper and can be omitted. 2-tress are already defined earlier.}
\fi

\section{Edge-length ratio of $2$-trees}\label{se:lower-bound}

We recall that $2$-trees are planar graphs. The main result of this section is the following.

\begin{theorem}\label{thm:2trees-unbounded}
    For any $r \ge 1$, there exists a $2$-tree $G$ with edge-length ratio $\rho(G) \ge r$.
\end{theorem}

To prove Theorem~\ref{thm:2trees-unbounded}, for a given $r$ we argue that a sufficiently large $2$-tree, drawn with the longest edge having length $r$, contains a triangle with area at most $\frac14$ (Corollary~\ref{cor:area}).
Then, inside this triangle of small area we build a sequence of triangles with perimeters decreasing by at least $1$ at every two steps (Lemmas~\ref{lem:case_analysis_1} and \ref{lem:case_analysis_2}), which results in a triangle with an edge of length less than 1.

We consider a special subclass $\G=\{G_0,G_1,\dots\}$ of $2$-trees with labeled vertices and edges constructed as follows: $G_0$ is the complete graph $K_3$ whose vertices and edges are given the label $0$.
The graph $G_{i+1}$ is obtained by adding five simplicial vertices to each edge of label $i$ of $G_i$. Each newly created vertex and edge gets label $i+1$. See Fig.~\ref{fig:iteration} for an example where the black vertices and edges have label $0$, the blue ones have label $1$, and the red ones have label $2$.

\begin{figure}[h]
    \centering
    \includegraphics[width=0.9\textwidth,page=2]{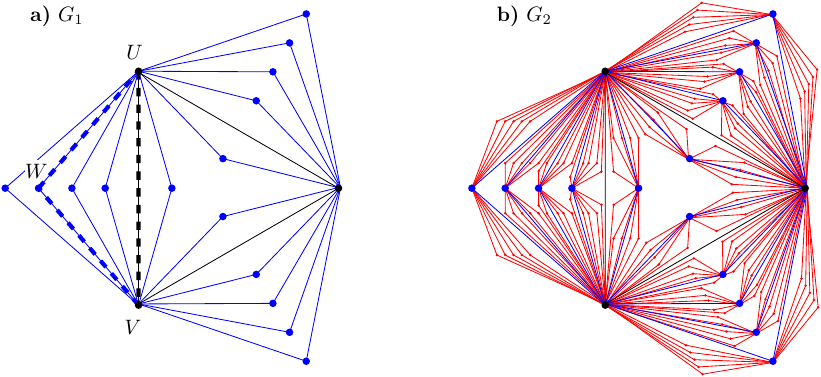}
    \caption{The $2$-trees $G_1$ and $G_2$. Black color corresponds to label $0$, blue to $1$, and red to $2$. Separating triangle $\Delta_1$ is emphasized by a dashed line in $G_1$.}%
    \label{fig:iteration}
\end{figure}

A \emph{separating triangle of level $i$} in a straight-line drawing of a $2$-tree $G$ is an unordered triple $\{U,V,W\}$ of mutually adjacent vertices such that the vertex $W$ of label $i$ was added as a simplicial vertex to the edge $UV$ in the recursive construction of $G$ and the triangle $\TRI UVW$
contains in its interior at least two other vertices with label $i$ which are simplicial to the edge $UV$.
%In particular, the triangle $\TRI UVW$ contains two vertices of $G$ with label $i$ in its interior.
For example, in Fig.~\ref{fig:iteration}~a) vertices $\{U,V,W\}$ form a separating triangle of level $1$.

\begin{lemma}\label{lem:separating-triangle}
    For any $k> i\ge 1$, for any planar straight-line drawing of the graph $G_k$, and for any edge $e$ of $G_k$ labeled by $i$, there exists a separating triangle of level $i+1$ containing the endpoints of $e$.

\end{lemma}
\begin{proof}%[of Lemma~\ref{lem:separating-triangle}]
    If a common edge of two triangles is traversed in the same direction when following their boundaries in the clockwise manner, then these triangles are nested, i.e. the interior of one contains the other one.
    Since we have five vertices simplicial to $e$, out of the corresponding five triangles in at least three
    $e$ traversed in the same direction when following their boundaries in the clockwise manner.
    Thus at least three triangles are nested and the outermost of these is the desired separating triangle.
\end{proof}

(For the clarity of presentation we have assumed a straight-line drawing, where the graph-theoretic
term triangle coincides with the geometric one. This assumption could indeed be neglected when we consider a triangle in a planar drawing as the Jordan curve formed from the drawing of a 3-cycle.)

We proceed to show that any drawing of $G_k$ contains a triangle of sufficiently small area. To this aim, we construct a sequence of nested triangles such that
each triangle's area is half of the previous triangle's area. We denote as $\Delta_i$ a separating triangle of level $i$ in an embedding of $G_k$, where $i \leq k$.

\begin{lemma}\label{lem:area}
For any $k\ge 1$, any planar straight-line drawing of $G_k$ contains a
sequence of triangles $\spread[\Delta,k]$, where for any $i\in\{1,\dots,k\}$
the triangle $\Delta_{i}$ is a separating triangle of level $i$,
and for each $i>1$, in addition, $\Delta_{i}$ is
in the interior of $\Delta_{i-1}$ and $A(\Delta_{i}) \leq \frac12 A(\Delta_{i-1})$.
\end{lemma}

\begin{proof}%[of Lemma~\ref{lem:area}]
We prove the lemma by induction on $i$.
For $i=1$ we apply Lemma~\ref{lem:separating-triangle} on any edge $e$ of label 0 in $G_k$
to get the triangle $\Delta_1$.

When $i\in \{2,\dots,k\}$, we assume
by inductive hypothesis that the graph $G_k$ contains a sequence of triangles $\spread[\Delta,i-1]$
satisfying the constraints.
Let $U$ be one of the two vertices of label $i-1$ in the interior of $\Delta_{i-1}$
and let $e$ and $f$ be the two edges of label $i-1$ incident to $U$.

We apply Lemma~\ref{lem:separating-triangle} on both of $e$ and $f$ to obtain two
separating triangles of level $i$ inside $\Delta_{i-1}$, see Fig.~\ref{fig:area_lemma}.
Since the drawing was planar, the two triangles are non-overlapping.
We choose the triangle with the smaller area to be $\Delta_i$ to assure that
$A(\Delta_i) \leq \frac12 A(\Delta_{i-1})$.

\begin{figure}
        \centering
        \includegraphics[width=0.9\textwidth]{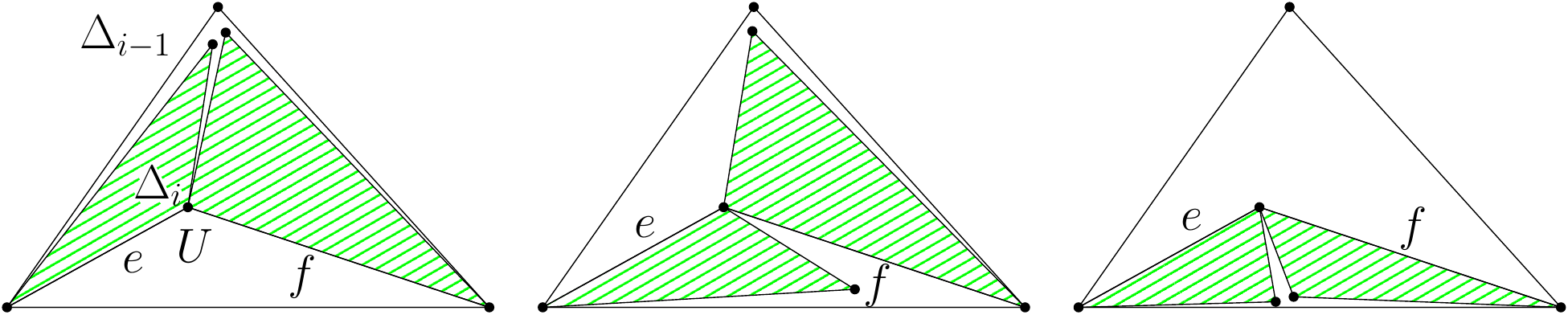}
        \caption{Two separating triangles created in the interior of $\Delta_{i-1}$}%
        \label{fig:area_lemma}
\end{figure}
\end{proof}

\begin{corollary}\label{cor:area}
For any $r > 1$
and $k \ge 2+ 2 \log_2 r$,
%and $k \ge 1+ \log_2(\sqrt{3}r^2)$,
every planar straight-line drawing of $G_k$
with edge lengths at most $r$
contains a separating triangle of area at most $\frac14$.
\end{corollary}

\begin{proof}
If all edges have length at most $r$, the area of $\Delta_1$ is bounded by $\frac{\sqrt{3}}{4}r^2$.
By Lemma~\ref{lem:area}, any drawing of $G_k$ contains a sequence of nested separating triangles whose last element $\Delta_k$ has area at most $\frac{1}{2^{k-1}} \frac{\sqrt{3} r^2}{4}\le \frac14$.
\end{proof}

Before we proceed to the next step in our construction, we need some elementary facts from the trigonometry.

We call \emph{thin} any triangle with edges of length at least $1$ and area at most $\frac14$.
Any thin triangle has height at most $\frac12$ and hence it has one obtuse angle of size at least $\frac{2\pi}3$ and two acute angles, each of size at most $\frac{\pi}6$.

\begin{lemma}\label{lem:angles}
Let $\TRI ABC$ be a thin triangle, where the longest edge is $AB$ and let $D\in \TRI ABC$
be such that $|CD|\ge 1$. Then one of the angles $\angle ACD$ or $\angle BCD$ is obtuse.
\end{lemma}

\begin{proof}
Assume by contradiction that both $\angle ACD$ and $\angle BCD$ are acute.
Without loss of generality we may also assume that $\angle ACD \ge \angle BCD$.
Since $\angle ACD + \angle BCD = \angle ACB \ge \frac{2\pi}3$, it follows that $\angle ACD \ge \frac{\pi}3$.

\begin{figure}
        \centering
        \includegraphics[width=0.8\textwidth]{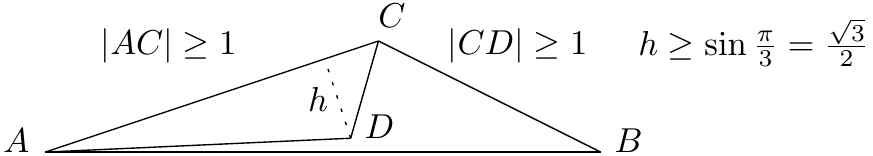}
        \caption{To the argument that $\angle ACD$ cannot be acute.}%
        \label{fig:angles}
\end{figure}

Then the triangle $\TRI ACD$ has height at least $\frac{\sqrt{3}}2$, see Fig.~\ref{fig:angles}.
Thus it has area at least $\frac{\sqrt{3}}4$, a contradiction with the fact that the surrounding thin triangle $\TRI ABC$ has area at most $\frac 14$.
\end{proof}

Now we focus our attention on the perimeters of the considered triangles.

\begin{lemma}\label{lem:cutting}
Let $\TRI ABC$ be a thin triangle, where the longest edge is $AB$.
Denote by $Q$ the polygon, created by cutting off an isosceles triangle $\TRI BDE$ with both edges $BD$ and $BE$ of length $1$. Then the perimeter of any triangle located in the polygon $Q$ is at most $P(\TRI ABC)-1$.
  %  \todo[inline]{R3: Lemma 4: wouldn't the same proof even yield $A(\Delta_i) < 1/4 A(\Delta_i-1)$ since you can take the smaller out of 4 triangles?}
%\todo[inline, color = green]{JF: I guess it relates to Lemma 2 and the answer is NO, since we can guarantee only two SEPARATING triangles, one per each of the edges $e$, $f$}
\end{lemma}

\begin{figure}
    \centering
    \includegraphics[width=0.75\textwidth]{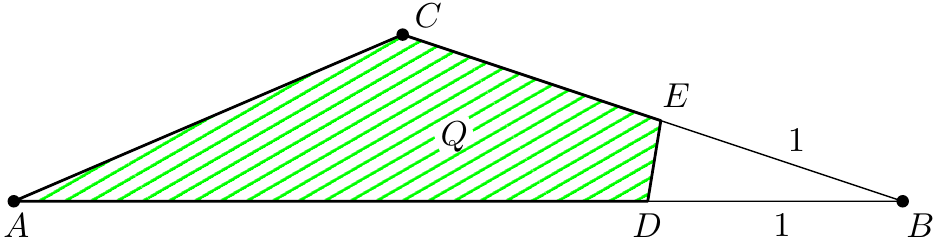}
    \caption{Cutting-off the triangle $\TRI BDE$.}%
    \label{fig:cutting}
\end{figure}
See Fig.~\ref{fig:cutting} for an example of cutting off an isosceles triangle.
\begin{proof}
Assume for a contradiction that some triangle $T$ has perimeter
$P(T)> P(\TRI ABC)-1$. Since $T$ and $Q$ are nested convex objects, we have that
that $P(Q) \ge P(T) > P(\TRI ABC)-1$. Then the length of the edge $DE$ is greater than $1$ and hence the angle $\angle DBE \ge \frac{\pi}{3}$, a contradiction with the property that the acute angles of a thin triangle are at most $\frac{\pi}{6}$.
\end{proof}

We now return to our construction and show that a separating triangle with a small area is guaranteed to contain a separating triangle of a significantly smaller perimeter. In the following two lemmas we distinguish two complementary cases, namely
whether the edge of level $i-1$ of a separating triangle of level $i$ is incident to its obtuse angle or not.

\begin{lemma}\label{lem:case_analysis_1}
Let $G_k$ have a planar straight-line drawing with edge lengths at least $1$ and let $\TRI UVW$ be a thin separating triangle of level $i$, where $i\le k-1$.
Assume that the edge $UV$ is of level $i-1$ and that it is incident to the obtuse angle of $\TRI UVW$.
Then $\TRI UVW$ contains a thin separating triangle $T$ of level $i+1$ whose perimeter satisfies $P(T) \leq P(\TRI UVW)-1$.
\end{lemma}

\begin{proof}
Let $X$ and $Y$ be the two vertices of level $i$ simplicial to the edge $UV$ inside the triangle $\TRI UVW$.
As the embedding of $G_k$ is non-crossing straight-line, we may assume without loss of generality that the vertex $X$ is inside $\TRI UVY$.

As all triangles in our further consideration are inside the thin triangle $\TRI UVW$, they have area at most $\frac14$. By the definition of thin triangle they are also thin, as otherwise we would get in $G_k$ an edge shorter than $1$, which violates the assumptions of the Lemma. We distinguish several cases depending on the position of the obtuse angle of the considered triangles, see Fig.~\ref{fig:case_analysis_1}.

\begin{figure}
    \centering
    \includegraphics[width=1.0\textwidth,page=1]{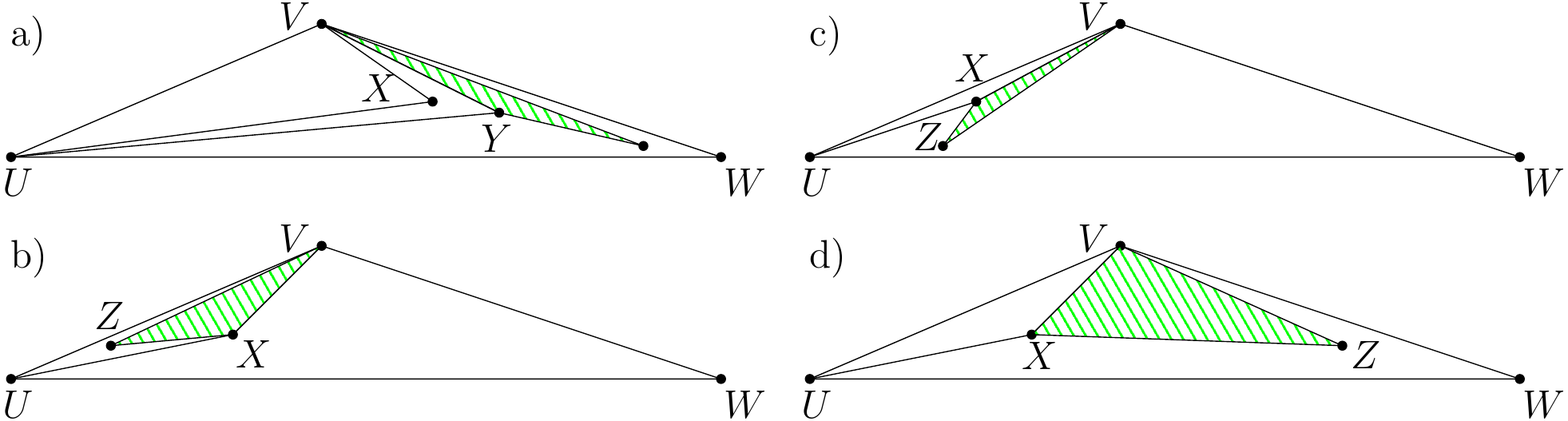}
    \caption{Case analysis for Lemma~\ref{lem:case_analysis_1}.}%
    \label{fig:case_analysis_1}
\end{figure}

\begin{itemize}
\item[a)] The obtuse angle of $\TRI UVX$ is at $V$.
By Lemma~\ref{lem:separating-triangle} we find a separating triangle $T$ incident with the edge $VY$.
Since $T$ takes place within the angle $\angle WVX$, it is at distance at least $1$ from $U$. Hence we may apply Lemma~\ref{lem:cutting} to cut away the isosceles triangle in the neighborhood of vertex $U$, to argue that the  perimeter of $T$ is at most $P(\TRI UVW)-1$.

\item[b)] The obtuse angle of $\TRI UVX$ is at $X$ and the separating triangle $\TRI VXZ$ incident with $VX$ obtained by by Lemma~\ref{lem:separating-triangle} is inside $\TRI UVX$. As $\TRI UXV$ is thin, we get that $\angle WVX \ge \frac{\pi}2$. Hence all points of $\TRI VXZ$ are at distance at least 1 from $W$.
%\todo[color=green]{Confirm that $\angle WVX$ is the correct angle, JF: yes}
We cut away the vertex $W$ and obtain $P(\TRI VXZ)\leq P(\TRI UVW)-1$.

\item[c)] The angle $\angle UXV$ is obtuse, the separating triangle $\TRI VXZ$ is outside $\TRI UVX$ and
the angle $\angle VXZ$ is obtuse. We apply Lemma~\ref{lem:angles} to get that $\angle WVZ$ is obtuse --- the case of $\angle UVZ$ being obtuse is excluded as this angle is composed from acute angles of two thin triangles: $\TRI UVX$ and $\TRI VXZ$. Then we cut away the vertex $W$ as in the previous case and obtain the claimed result.

\item[d)] The angle $\angle UXV$ is obtuse, the separating triangle $\TRI VXZ$ is outside $\TRI UVX$ and
the angle $\angle VXZ$ is acute. Now cut away the vertex $U$ (as $|UX|\ge 1$), and get $P(\TRI VXZ)\leq P(\TRI UVW)-1$.
\end{itemize}
\end{proof}

Note that \emph{only} when Case a) occurred, we used the existence of two vertices of label $i$ within the separating triangle $\TRI UVW$. If $Y$ was not present, we would have to discuss the case that the obtuse angle of $\TRI UVX$ is at $V$ and both separating triangles of level $i+1$ are inside $\TRI UVX$. For such a case it is possible to find a configuration where Lemma~\ref{lem:cutting} cannot be immediately applied, see Fig.~\ref{fig:problematic}.

\begin{figure}
    \centering
    \includegraphics[width=0.75\textwidth,page=2]{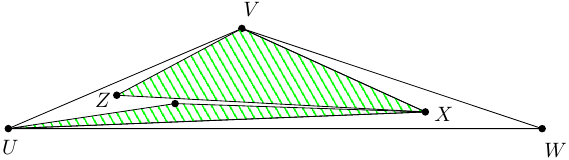}
    \caption{The case that avoids cutting. Note that $X$ could be arbitrary close to $W$ and $Z$ to $U$.}%
    \label{fig:problematic}
\end{figure}

\begin{lemma}\label{lem:case_analysis_2}
Let $G_k$ have a planar straight-line drawing with edge length at least $1$ and let $\TRI UVW$ be a thin separating triangle of level $i\le k-2$.
Assume that the edge $UV$ is of level $i-1$ and that it is \emph{not incident} to the obtuse angle of $\TRI UVW$.
Then $\TRI UVW$ contains a thin separating triangle $T$ of level at most $i+2$ whose perimeter satisfies $P(T) \leq P(\TRI UVW)-1$.
\end{lemma}

\begin{proof}
Similarly to the previous lemma, let $X$ be one of the two vertices of level $i$ simplicial to the edge $UV$ inside the triangle $\TRI UVW$, see Fig.~\ref{fig:case_analysis_2}. By Lemma~\ref{lem:separating-triangle} we construct a separating triangle $\TRI UXZ$ incident with the edge $UX$.
\begin{figure}
    \centering
    \includegraphics[width=0.8\textwidth,page=5]{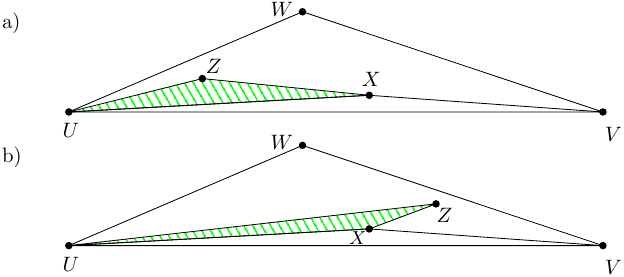}
    \caption{Case analysis for Lemma~\ref{lem:case_analysis_2}.}%
    \label{fig:case_analysis_2}
\end{figure}
\begin{itemize}
\item[a)] If the angle $\angle UXZ$ is acute, then we cut away $V$ and apply Lemma~\ref{lem:cutting} to obtain $P(\TRI UXZ) \leq P(\TRI UVW)-1$.
\item[b)] If the angle $\angle UXZ$ is obtuse, then we apply Lemma~\ref{lem:case_analysis_1} for the triangle $\TRI UXZ$ to find a suitable separating triangle $T$ of level $i+2$ within $\TRI UXZ$.
\end{itemize}\end{proof}

\begin{corollary}\label{cor:perimeter}
For any $r > 1$, $k\ge 1$, $l\ge 0$ and any planar straight-line drawing of $G_{k+l}$ with edge length at least $1$
it holds: If the drawing contains a thin separating triangle of level $k\ge 1$, then it has a triangle of perimeter at most $2r+\frac14-\lfloor \frac{l}2 \rfloor$.
\end{corollary}

\begin{proof}
Denote by $\Delta_0$ the thin triangle of level $k$ in the drawing of $G_{k+l}$.
Since all edges have length at most $r$, any thin triangle it could be drawn inside a rectangle $r\times \frac18$, hence it has perimeter at most $2r+\frac14$.

We involve Lemmas~\ref{lem:case_analysis_1} and \ref{lem:case_analysis_2}, to find in the drawing of $G_{k+l}$
a sequence of nested separating triangles of length at least $l+1$ with decreasing perimeters.

We argue that the sequence can be chosen such that
for any $i\in\{1,2,\dots,\lfloor \frac{l}2 \rfloor\}:
P(\Delta_{2i})\le P(\Delta_{2i-2})-1 \le P(\Delta_0)-i$.
We distinguish two cases whether the edge of level $2i-3$ in $\Delta_{2i-2}$ is incident to the obtuse angle of $\Delta_{2i-2}$ or not:
\begin{itemize}
\item
In the first case we apply Lemma~\ref{lem:case_analysis_1} to get $P(\Delta_{2i-1})\le P(\Delta_{2i-2})-1$. As $\Delta_{2i}$ is inside $\Delta_{2i-1}$, we get $P(\Delta_{2i})\le P(\Delta_{2i-2})-1$.
\item Otherwise we apply Lemma~\ref{lem:case_analysis_2} to derive $P(\Delta_{2i})\le P(\Delta_{2i-2})-1$ directly.
\end{itemize}
\end{proof}

Now we combine the two parts together to prove Theorem~\ref{thm:2trees-unbounded}.

\begin{proof}[of Theorem~\ref{thm:2trees-unbounded}]
For given $r$ we choose
$k = \lceil 2+ 2 \log_2 r \rceil$
and consider the graph $G_{k+4r}$. Assume for a contradiction that $G_{k+4r}$ allows a drawing of edge-length ratio at most $r$. Up to an appropriate scaling, we assume that the longest edge of such drawing has length $r$ and hence the shortest has length at least $1$.

In the drawing of the graph $G_{k+4r}$ consider a sequence of separating triangles $\Delta_1,\dots,\Delta_{k+4r}$ where $\Delta_1,\dots,\Delta_k$ are chosen as shown in Corollary~\ref{cor:area}.

By Corollary~\ref{cor:area}, the triangle $\Delta_k$ is thin, so we can extend the
sequence with $\Delta_k,\dots,\Delta_{k+4r}$ according to Corollary~\ref{cor:perimeter}.

By Corollary~\ref{cor:perimeter}, $P(\Delta_{k+4r})\le 2r+\frac14 -2r = \frac14$, a contradiction to the assumption that all triangles of $G_{k+4r}$ have sides of length at least one.
\end{proof}

Note that the graph $G_{k+4r}$ has $O^*\big((10^4)^r\big)$ vertices and edges, as in each iteration we add 10 edges of level $i$ per every edge of level $i-1$. The dependency between the edge-length ratio and the number of vertices could be rephrased as follows:

\begin{corollary}\label{cor:ratio-by-n}
The edge-length ratio over the class of $n$-vertex $2$-trees is $\Omega(\log n)$.
\end{corollary}

We recall that Borrazzo and Frati prove that every partial $2$-tree with $n$ vertices
admits a planar straight-line drawing whose edge-length ratio is in $O(n^{0.695})$~\cite[Corollary~1]{BorrazzoF20}.

\section{Local edge-length ratio of $2$-trees}\label{se:local}

The aesthetic criterion studied in the previous section took into account any pair of edges. By our construction of nested triangles, it might happen that two edges attaining the maximum length ratio are far
in the graph distance (in the Euclidean distance they are close as the triangles are nested). This observation leads us to the question,
whether $2$-trees allow drawings where the length ratio of any two adjacent edges could be bounded by a constant. For this purpose we define the local variant of the edge-length ratio as follows:

The \emph{local edge-length ratio} of a planar straight-line drawing of a graph $G$ is the maximum ratio between the lengths of two adjacent edges (sharing a common vertex) of the drawing.

\begin{definition}\label{def:local}
The \emph{local edge-length ratio} $\rho_l(G)$ of a planar graph $G$ is the infimum local edge-length ratio taken over all planar straight-line drawings of $G$.

$$
\rho_l(G) = \inf_{\text{drawing of }G} \quad \max_{UV,VW\in E_G} \frac{|UV|}{|VW|}
$$
\end{definition}

Observe that the local edge-length ratio $\rho_l(G)$ is by definition bounded by the global edge-length ratio $\rho(G)$.
In particular, every outerplanar graph $G$ allows a drawing witnessing $\rho_l(G)\le 2$~\cite{LazardLL19}.
We extend this positive result to the class of all $2$-trees with a slightly increased bound on the ratio.

\begin{theorem}\label{thm:local}
 The local edge-length ratio of any $n$-vertex $2$-tree $G$ is $\rho_l(G)\le 4$. Also, for any arbitrarily small positive constant $\varepsilon$, a planar straight-line drawing of $G$ with local edge-length ratio at most $4 +\varepsilon$ can be computed in $O(n)$ time assuming the real RAM model of computation.
\end{theorem}

The proof of Theorem~\ref{thm:local} is based on a construction that provides a straight-line drawing of local edge-length ratio $4+\varepsilon$ for any given $2$-tree $G$ and any $\varepsilon>0$.

We use a breadth first search (BFS) and
and decompose $V_G$ into layers based on the distance from the initial edge $e$ of the recursive definition of the $2$-tree. Each such layer $L_i=\{u: \dist(u,e)=i\}$ is a forest, see Fig.~\ref{fig:local_example} a).

\begin{figure}
    \centering
    \includegraphics[page=1]{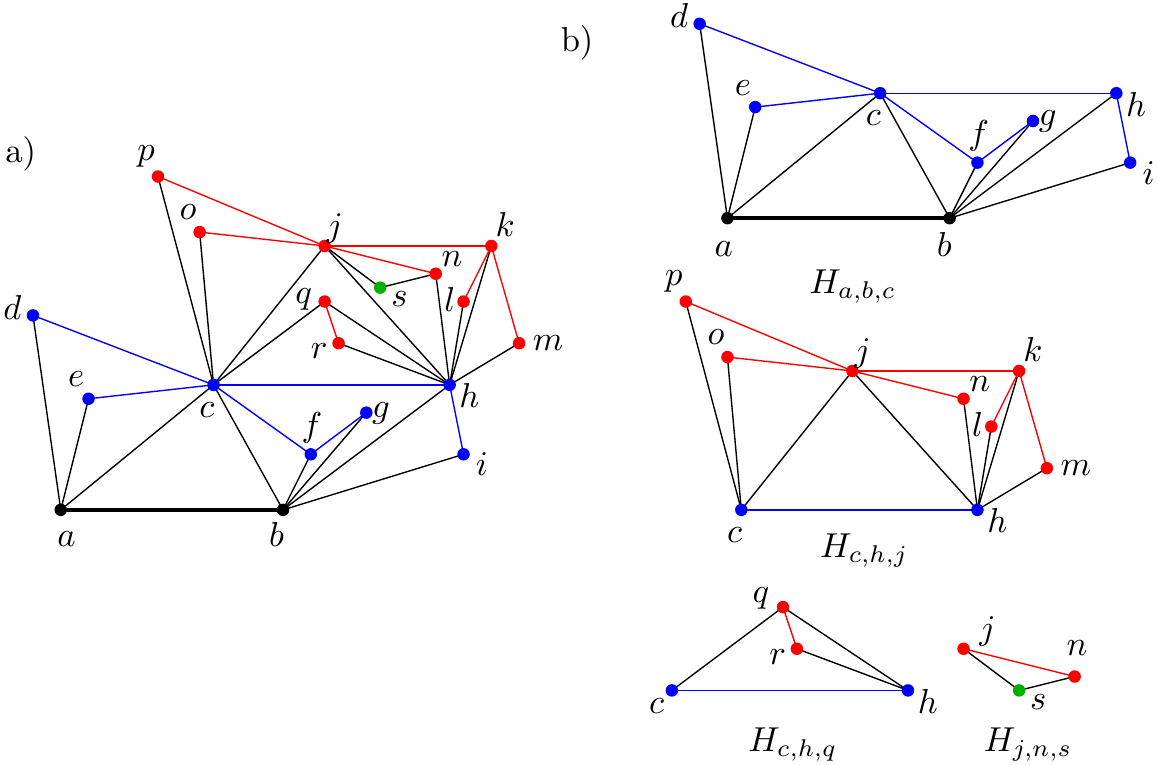}
    \caption{a) A decomposition of a $2$-tree $G$ into layers: black $L_0$ (the initial edge), blue $L_1$, red $L_2$, and green $L_3$;
    b) The tree components of $G$}%
    \label{fig:local_example}
\end{figure}

Moreover, for every
component $C$ of $L_i, i\ge 1$ we may due to the definition of a $2$-tree identify a unique vertex $w\in C$
and two its neighbors $u,v\in L_{i-1}$, as $w$ is the first vertex of $C$ inserted into $G$, and in the time of its insertion it was simplicial to the edge $uv$. We call the subgraph of $G$ induced by $C\cup\{u,v\}$ a \emph{tree-component} rooted in $u,v$ and denote it by $H_{u,v,w}$, see Fig.~\ref{fig:local_example} b). Observe that each tree-component of itself is a $2$-tree.
Moreover the vertices of $H_{u,v,w}$ distinct from $u,v$ and $w$ can be partitioned into two disjoint sets: those adjacent to $u$ and those adjacent to $v$.

Note that BFS can be executed in $O(n)$ time for a planar graph with $n$ vertices. This procedure can be extended
in a straightforward way to determine the tree-components in $O(n)$ time --- on each vertex we spend additional
constant time to identify the tree component it belongs.

For a line segment $AB$, let $\overline{AB} = AB\setminus\{A,B\}$ denote for the segment $AB$ without its endpoints.

\begin{definition}
Let $UV$ be an edge of a planar straight-line drawing of $G$ on at least three vertices.
The \emph{vacant region} for $UV$ is the intersection of all open half-planes determined by all pairs of vertices such that these half-planes contain $\overline{UV}$.
\end{definition}

For example, Fig.~\ref{fig:vacant} shows the vacant region for an edge
$UV$ in a planar straight-line drawing of a $2$-tree. Note that, by the definition, the vacant region for $UV$ is an open convex set with $U$ and $V$ on the boundary.
\begin{figure}
    \begin{center}
    \includegraphics{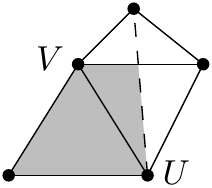}
    \caption{The filled gray region is the vacant region of the edge $UV$.}%
    \label{fig:vacant}
    \end{center}
\end{figure}

We proceed to the main technical step of our construction.

\begin{lemma}\label{lem:local}
Let $H_{X,Y,Z}$ be a tree-component of a $2$-tree $G$. For any $\delta>0$, any open convex set $S$ and any two points on the boundary of $S$, the graph $H_{X,Y,Z}$ can be drawn with the local edge-length ratio at most $2+\delta$ such that vertices $X,Y$ are placed on the chosen two points, the rest of the drawing of $H_{X,Y,Z}$ is inside $S$, and $XY$ is the longest edge of the drawing.
\end{lemma}

\begin{figure}
    \begin{center}
    \includegraphics{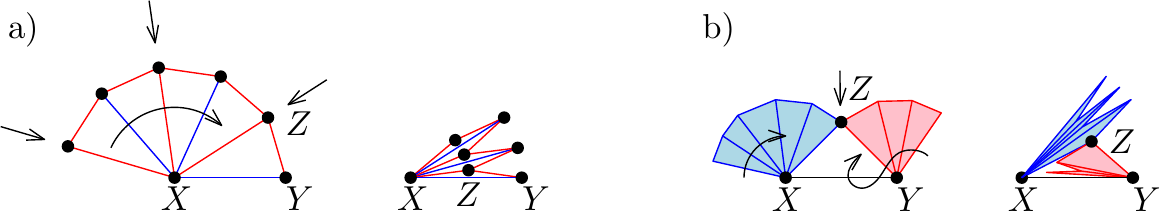}
    \caption{Folding a path in the tree component $H_{X,Y,Z}$. The arrows indicate the vertex movement.}%
    \label{fig:folding}
    \end{center}
\end{figure}

\ifarxiv
We provide a very brief idea of the construction used in the proof of Lemma~\ref{lem:local}.
The full proof is present in the appendix.
\else
Instead of the proof of Lemma~\ref{lem:local}, that is omitted due to space restrictions, we provide a very brief idea of the construction.
\fi
Observe that in the case that in the case when the tree component $H_{X,Y,Z}$ is a fan centered at $X$, then it can be folded like an umbrella into the vacant region of $XY$ as depicted in Fig.~\ref{fig:folding} a). In the folded drawing the red edges have the same length upto an additive factor $\delta$, while the blue are twice longer (again upto $+\delta$).

Analogously, if the vertices adjacent to $X$ in $H_{X,Y,Z}$ induce a path and the same for the neighbors of $Y$, then these two paths can be folded from both sides of $XY$ inside its vacant region, see Fig.~\ref{fig:folding} b). By much more technically involved argument it can be shown that the whole branch of a tree can be folded into the area near the first edge of the branch.

\begin{proof}[of Theorem~\ref{thm:local}]
When $G=K_2$, then it has $\rho_l(K_2)=1$, by Definition~\ref{def:local} (note that $U$ and $W$ need not to be distinct.) Otherwise we proceed by induction on the number of tree components of $G$.

For any $\varepsilon\in(0,1)$, let $\delta = \frac{\varepsilon}3$.
The induction hypothesis we aim to prove is:
\begin{claim}
Any $2$-tree $G$ allows a drawing with local edge-length ratio at most $4+\varepsilon$, where each tree component $H_{X,Y,Z}$ is drawn with local edge length ratio at most $2+\delta$ and $XY$ is the longest edge of the drawing of $H_{X,Y,Z}$.
\end{claim}

For the base of the induction $G$ consists of a single tree component $H_{X,Y,Z}$, where
$XY$ is the initial edge of construction of $G$ as a $2$-tree.
We choose any open convex set $S$ and two points $X$, $Y$ on its boundary and apply
Lemma~\ref{lem:local}.

For the induction step assume that $H_{X,Y,Z}$ is a tree component of $G$, where $Z$ belongs to the highest possible level.
The graph $G'=G\setminus (H_{X,Y,Z}\setminus \{X,Y\})$ (i.e. when we remove from $G$ the component of $L_i$ containing the vertex $Z$) is a $2$-tree, since we may create $G'$ as a $2$-tree by the same order of insertions as is used for $G$, only restricted to the vertices of $G'$.
By induction hypothesis $G'$ allows a drawing with local edge-length ratio at most $4+\varepsilon$.

In this drawing we identify the vacant region $S$ for $XY$ and involve Lemma~\ref{lem:local} to extend the drawing
of $G'$ to the entire $G$.
The only vertices common to $G'$ and $H_{X,Y,Z}$ are $X$ and $Y$, hence we shall argue that edges incident with $X$ or $Y$ have edge-length ratio at most $4+\varepsilon$, as inside $H_{X,Y,Z}$ the ratio is at most $2+\delta < 4+\varepsilon$ by Lemma~\ref{lem:local}.

By the construction of the $2$-tree, the edge $XY$ may belong to several tree components rooted in $X$, $Y$, where it is the longest edge, but only to a single tree-component rotted in the vertices of the preceding level.
Consequently, the edge-length ratio of any two edges incident with $X$ or with $Y$ is at most $(2+\delta)^2=4+2\delta+\delta^2<4+3\delta =4+\varepsilon$.

Finally, we remark that computing the coordinates of the vertices can be executed in constant time per vertex, assuming a real RAM model of computation. It follows that the drawing of $G$ can be computed in $O(n)$ time.
\end{proof}

Since any graph of treewidth at most $2$, in particular all series-parallel graphs, can be augmented to a $2$-tree, Theorem~\ref{thm:local} directly implies the following.

\begin{corollary}\label{co:local}
For any graph $G$ of treewidth at most $2$, it holds that $\rho_l(G)\le 4$.
\end{corollary}

\section{Conclusions and Open Problems}\label{se:open}

This paper studied the edge-length ratio of planar straight-line drawing of partial $2$-trees.
It proved an $\Omega(\log n)$ lower bound on such edge-length ratio and it proved that every partial $2$-tree admits a planar straight-line drawing such that the local edge-length ratio is at most $4 + \varepsilon$ for any arbitrarily small positive $\varepsilon$.
Several questions are naturally related with our results.
We conclude the paper by listing some of those that, in our opinion, are among the most interesting ones.

\begin{enumerate}

\item Corollary~\ref{cor:ratio-by-n} of this paper gives a logarithmic lower bound while Corollary~1 of~\cite{BorrazzoF20} gives a sub-linear upper bound on the edge-length ratio of planar straight-line drawings of partial $2$-trees. We find it interesting to close the gap between the upper and lower bound.

\item Theorem~\ref{thm:local} gives an upper bound of $4$ on the local edge-length ratio of partial $2$-trees. It would be interesting to establish whether such an upper bound is tight. Also, studying the local edge-length ratio of other families of planar graphs is an interesting topic.

\item The construction in Theorem~\ref{thm:local} creates drawings where the majority of angles are very close to $0$ or $\pi$ radians. Hence, it would make sense to study the interplay between (local or global) edge-length ratio and angular resolution in planar straight-line drawings.

%\item Though our study was motivated by aesthetic graph drawing, we are aware that the drawings obtained by our method would not be well readable. The reason is that often the edges are drawn close to each other, in other words, the faces would often have a very small area. We leave for further investigations whether imposing additional constraints, such as a lower bound on the angle between adjacent edges for bounded-degree graphs, would provide drawings that are more satisfactory from the human-reading perspective and what would the edge-length ratio bounds be in such a scenario.

\end{enumerate}

\section*{Acknowledgement}
We thank all three reviewers for their positive comments and careful review, which helped improve our contribution.

% ------------------------------------------------------------

\bibliographystyle{splncs04}
\bibliography{GDsubmission}

% \end{document}

\clearpage

\ifarxiv%== ARXIV =========================================================

\appendix
%\section*{\LARGE Appendix}\label{se:appendix}\vspace{1em}
\makeatletter
\noindent
\rlap{\color[rgb]{0.51,0.50,0.52}\vrule\@width\textwidth\@height1\p@}%
\hspace*{7mm}\fboxsep1.5mm\colorbox[rgb]{1,1,1}{\raisebox{-0.4ex}{\large\selectfont\sffamily\bfseries Appendix}}%
\makeatother

\section{Proof of Lemma~\ref{lem:local}}

\begin{proof}Up to an appropriate scaling, we may assume
that $|XY|=2+\delta$ and hence the goal is to design a drawing where
all edges have length at least 1 and less than $2+\delta$.

Recall that $H_{X,Y,Z}$ consists of $X,Y$ and a tree $T$ containing $Z$. We draw $H_{X,Y,Z}$, such that
the length of any edge $XU$ or $YU$ will belong to the interval $(2,2+\delta)$ if and only if the distance between $U$ and $Z$ in $T$ is odd.
We call such edges/segments \emph{long}. All other edges will be drawn to have length from the interval $(1,1+\delta)$. We call these \emph{short}.

We distinguish several cases. If $H_{X,Y,Z}\simeq K_3$, it suffices to draw $Z$ inside $S$ such that both $XZ$ and $YZ$ are short. A suitable position for such $Z$ close enough to the center of $XY$ always exists.

Otherwise assume without loss of generality that vertices $Y,Z$ have a common neighbor $X'$ in $H_{X,Y,Z}$. We first find a position for $X'$ in $S$ sufficiently close to $X$, such that $X'Y$ is long. Then we determine the position of $Z$ inside the triangle $XYX'$, such that all edges incident at this moment with $Z$ are short, see Fig.~\ref{fig:initriangle}.

\begin{figure}
    \centering
    \includegraphics[width=0.6\textwidth,page=1]{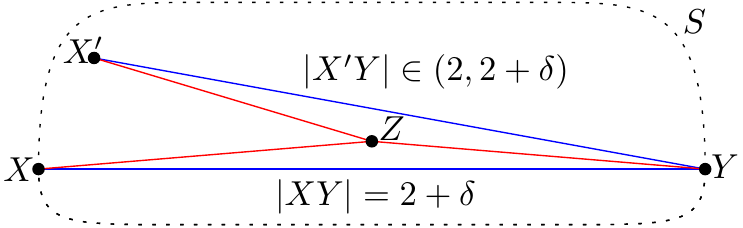}
    \caption{The initial position of points $X,Y,Z$ and $X'$. The boundary of $S$ is indicated by dots. All red edges are short, i.e. of length in $(1,1+\delta)$. The long edges are blue.}%
    \label{fig:initriangle}
\end{figure}

From now on, all vertices adjacent to $X$ will be drawn inside $\TRI XYZ$, while neighbors of $Y$ inside $\TRI YX'Z$.
We show how to draw the neighborhood of $X$; for the neighbors of $Y$ the argument is almost identical, only the $X'$ will take the role of $X$ in this case.

We process the neighbors of $X$ by the order of increasing distance from $Z$ in the tree $T$.
Let $U$ be a neighbor of $X$ and let $V,W$ be the first two vertices along the path from $U$ to $Z$ in $T$. (If $V=Z$, we choose $Y$ for $W$.) By our assumption on the order in which the vertices are processed, we know that we have already drawn all common neighbors of $X$ and $W$ but, possibly, only some of the neighbors of $V$.

\begin{itemize}
\item[a)] If the edge $XV$ is short then determine the already drawn vertex $U'\in\TRI XVW, U'\ne X,V$ such that
$\angle VXU'$ is minimal.
If $U'\in N(V)$ then the triangle $\TRI XVW$ contains no other vertex drawn so far and we may put $U$ near $U'$ so that $XU$ is long (as well as $XU'$) and $VU$ short (like $VU'$). Note that this case also covers the situation when $U'=W$.

When $U'\notin N(V)$, we shall first exclude the case $U'\in N(V')$ for some $V'\in N(W)$. Since $U'$ would be drawn inside $\TRI XWV'$ by our method, we would get $\angle VXU' > \angle VXV'$, a contradiction with the choice of $U'$.

Hence it means that $U'\in N(W)$. By the choice of the position of $U'$ as described in the Case~(b) just below, it is also possible to draw $\TRI XVU$ such that $XU$ is long and $VU$ short.

\item[b)] If the edge $XV$ is long, then first determine an auxiliary point $A$ that is in the vacant region for $XV$ such that $|XA|>2$ and $|AV'|>1$ for all $V'\in N(V)$. We then draw $U$ on the segment $XA$ such that both edges $XU$ and $UV$ are short.

    Note first that it is the position of the point $A$ inside the vacant region that allows us to find a position of $U$ for the Case~(a) above, i.e., at least 2 units far from $X$ and at least 1 unit far from $V$ (in the Case~(b) denoted as $V'$), see Fig.~\ref{fig:initriangle2}.

Also note that the concept of vacant regions forces the drawing of common neighbors of $X$ and $V$ to be performed in a nested way, i.e., only the first drawn neighbor of $X$ and $V$ will affect the position of $U'$ as the further such neighbors will be drawn inside $\TRI XVU$.

\begin{figure}
\begin{center}
\iffalse
    \centering
    \begin{subfigure}{0.47\textwidth}
        \caption{}\label{fig:initriangle-a}
        \includegraphics[width=1.0\textwidth,page=2]{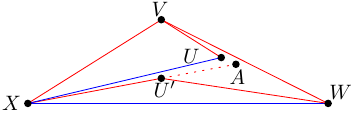}%
    \end{subfigure}
    \hfil
    \begin{subfigure}{0.47\textwidth}
        \caption{}\label{fig:initriangle-b}
        \includegraphics[width=1.0\textwidth,page=3]{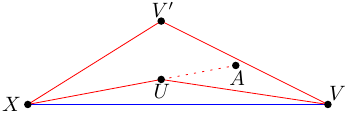}%
    \end{subfigure}
\fi
    \includegraphics[width=0.9\textwidth,page=4]{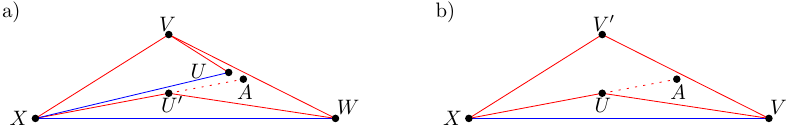}%
    \caption{Finding the position of $U$ in the case~b), accompanied by the case~a). Note that the same vertices are denoted differently in distinct iterations of the process.}%
    \label{fig:initriangle2}
\end{center}
\end{figure}
\end{itemize}
\end{proof}

\fi%== END ARXIV =========================================================

\end{document}